\pdfoutput=1
\documentclass[letterpaper, 12pt]{article}

\usepackage[margin=1in,top=1in,bottom=1in]{geometry}
\usepackage{textcomp}
\usepackage{setspace}
\usepackage{amssymb,amsmath,amsthm}
\usepackage{epsf}
\usepackage{times,bm,mathrsfs}
\usepackage{amsmath}
\usepackage{amssymb}
\usepackage{amsfonts}
\usepackage{mathrsfs}
\usepackage{bbm}
\usepackage{color}
\usepackage{graphicx}
\usepackage{amscd}
\usepackage{epsfig}

\allowdisplaybreaks

\definecolor{Red}{rgb}{1,0,0}
\definecolor{Blue}{rgb}{0,0,1}
\definecolor{Olive}{rgb}{0.41,0.55,0.13}
\definecolor{Green}{rgb}{0,1,0}
\definecolor{MGreen}{rgb}{0,0.8,0}
\definecolor{DGreen}{rgb}{0,0.55,0}
\definecolor{Yellow}{rgb}{1,1,0}
\definecolor{Cyan}{rgb}{0,1,1}
\definecolor{Magenta}{rgb}{1,0,1}
\definecolor{Orange}{rgb}{1,.5,0}
\definecolor{Violet}{rgb}{.5,0,.5}
\definecolor{Purple}{rgb}{.75,0,.25}
\definecolor{Brown}{rgb}{.75,.5,.25}
\definecolor{Grey}{rgb}{.5,.5,.5}
\definecolor{Black}{rgb}{0,0,0}

\title{On the variance of internode distance under the multispecies coalescent}
\author{
	S\'ebastien Roch\footnote{Department of Mathematics, University of Wisconsin--Madison, Madison WI 53706} 
}
\date{\today}

\newtheorem{thm}{Theorem}

\newtheorem{lemma}{Lemma}

\newcommand{\tree}{\mathcal{T}}

\newcommand{\prob}{\mathbf{P}}

\newcommand{\stree}{\mathcal{S}}

\definecolor{Red}{rgb}{1,0,0}
\definecolor{Blue}{rgb}{0,0,1}
\definecolor{Grey}{rgb}{0.5,0.5,0.5}

\newcommand{\inode}[1]{\hat\delta_{\mathrm{int}}^{#1}}
\newcommand{\einode}{\bar\delta_{\mathrm{int}}}
\newcommand{\ustree}{\stree_{\mathrm{u}}}

\newcommand{\var}{\mathbf{Var}}
\newcommand{\expec}{\mathbf{E}}

\newcommand{\fcoal}{\mathcal{F}}

\begin{document}	
	
\maketitle

\begin{abstract}
We consider the problem of estimating species trees from unrooted gene tree
topologies in the presence of incomplete lineage sorting,
a common phenomenon that creates gene tree heterogeneity
in multilocus datasets.
One popular class of reconstruction methods in this setting 
is based
on internode distances, i.e.~the average graph distance between pairs of species across gene trees. While statistical
consistency in the limit of large numbers of loci has been
established in some cases, little is known about the sample complexity of such methods. Here we make progress 
on this question by deriving a lower bound on the worst-case
variance of internode distance which depends linearly on
the corresponding graph distance in the species tree. We also discuss
some algorithmic implications.

\end{abstract}

\section{Introduction}

Species tree estimation is increasingly based on large numbers
of loci or genes across many species. Gene tree heterogeneity,
i.e.~the fact that different genomic regions may be consistent with incongruent
genealogical histories, is a common phenomenon in 
multilocus datasets that leads to significant challenges in this type of estimation. One important source of incongruence is
incomplete lineage sorting (ILS), a population-genetic effect
(see Figure~\ref{fig:ils} below for an illustration), which is
modeled mathematically using the multispecies coalescent (MSC)
process~\cite{Maddison1997a,Rannala2003}.  
Many recent phylogenetic analyses of genome-scale biological datasets have indeed revealed substantial heterogeneity consistent with ILS~\cite{jarvis-2014b,wickett2014,Cannon2016}.

Standard
methods for species tree estimation
that do not take this heterogeneity into account,
e.g.~the concatenation of genes followed by a 
single-tree maximum 
likelihood analysis, have been shown to suffer serious
drawbacks under the MSC~\cite{RochSteel2015,RoNuWa:18}.
On the other hand, new methods have been developed for species tree estimation
that specifically address gene tree heterogeneity.
One popular class of methods, often referred to as summary methods,
proceed in two steps: first reconstruct a gene tree for each locus; then infer a species tree from this collection of
gene trees. Under the MSC, many of these methods have been proven to converge to the true species tree
when the number of loci increases, i.e.~the methods are said to be statistically consistent. 
Examples of summary methods that enable statistically
consistent species tree estimation include
MP-EST~\cite{Liu2010a}, NJst~\cite{njst}, ASTRID~\cite{astrid}, ASTRAL~\cite{astral,astral-2},
STEM~\cite{stem}, STEAC~\cite{star}, STAR~\cite{star},
and
GLASS~\cite{Mossel2010}.

Here we focus on reconstruction methods, such as NJst and
ASTRID, based on what is known as internode distances,
i.e.~the average of pairwise graph distances across genes.
Beyond statistical consistency~\cite{njst,kreidl2011,AlDeRh:18}, little is known about
the data requirement or sample complexity of such methods
(unlike other methods such as ASTRAL~\cite{ShRoMi:18}
or GLASS~\cite{Mossel2010} for instance). That is, how many
genes or loci are needed to ensure that the true species tree
is inferred with high probability under the MSC?
Here we make progress 
on this question by deriving a lower bound on the worst-case
variance of internode distance. Indeed
the sample complexity of a reconstruction method depends
closely on the variance of the quantities it estimates, in this 
case internode distances.
Our bound depends linearly on
the corresponding graph distance in the species tree which, 
as we explain below, 
has possible implications for the choice of an accurate reconstruction
method.

The rest of the paper is structured as follows.
In Section~\ref{sec:basics}, we state our main results formally, after defining the MSC and the internode distance. In Section~\ref{sec:discussion}, we discuss algorithmic
implications of our bound. Proofs can be found in Section~\ref{sec:bad}.

\section{Definitions and results}
\label{sec:basics}

In this section, we first introduce the multispecies coalescent. 
We also define the internode distance and
state our results formally.

\paragraph{Multilocus evolution under the multispecies coalescent}
Our analysis is based on the multispecies coalescent (MSC),
a standard random gene tree model~\cite{Maddison1997a,Rannala2003}.
See Figure~\ref{fig:ils} for an illustration.
\begin{figure}
	\centering
	\includegraphics[width = 0.6\textwidth]{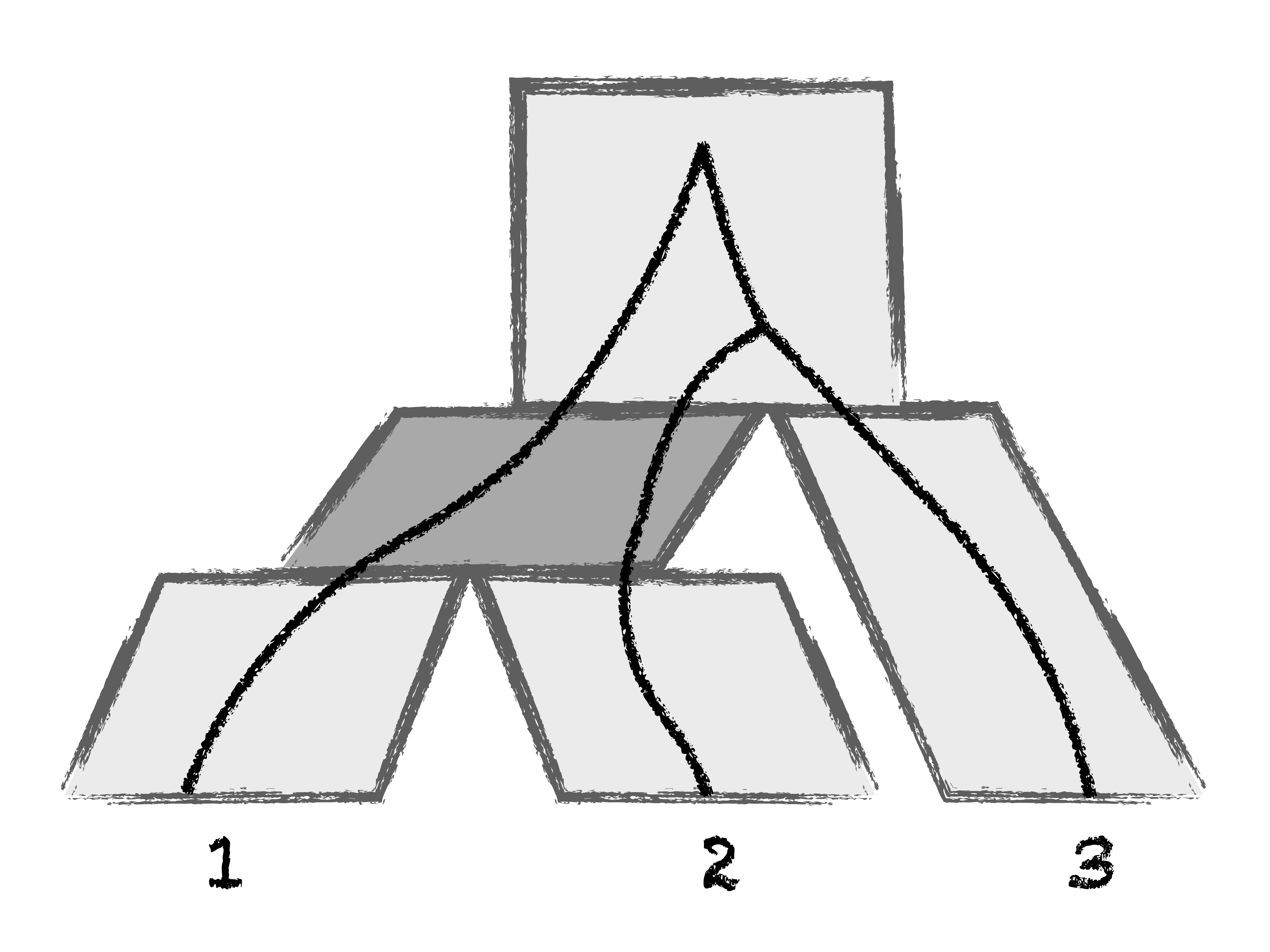}
	\caption{An incomplete lineage sorting event (in the rooted setting). Although $1$ and $2$ are more closely related in the rooted species tree (fat tree), $2$ and $3$ are more closely related in the rooted gene tree (thin tree). This incongruence is caused by the failure of the lineages originating
		from $1$ and $2$ to coalesce within the shaded branch. The shorter this branch is, the more likely incongruence occurs.}\label{fig:ils}
\end{figure}
Consider a \emph{species tree} 
$(\stree,\Gamma)$
with 
$n$ leaves. Here $\stree  = (\mathcal{V},\mathcal{E},r)$ is a rooted binary tree with vertex and edge sets $\mathcal{V}$ and $\mathcal{E}$ and where each leaf is labeled by a species in $\{1,\ldots,n\}$. We refer to $\stree$ as the \emph{species tree topology}. The \emph{branch lengths} $\Gamma = (\Gamma_e)_{e\in E}$
are expressed in so-called coalescent time units. We do not assume that $(\stree,\Gamma)$ is ultrametric (see e.g.~\cite{Steel:16}). Each gene\footnote{In keeping with much of the literature on the MSC, we use the generic term \emph{gene} to refer to any genomic region experiencing low rates of internal recombination, not necessarily a protein-coding region.} $j = 1,\ldots, m$ 
has a genealogical history 
represented by its gene tree $\tree_j$
distributed
according to the following process: 
looking backwards in time, 
on each branch $e$ of the species tree, 
the coalescence of any two lineages 
is exponentially distributed with rate $1$, 
independently from all other pairs; 
whenever two branches merge in the species tree, 
we also merge the lineages of the corresponding populations, that is, the coalescent proceeds on the union of the lineages;
one individual is sampled at each leaf. The genes are assumed to
be unlinked, i.e.~the process above is run independently
and identically for all $j = 1,\ldots,m$.
More specifically, the probability density of a realization of this model for $m$ independent genes is
\begin{align*}
&\prod_{j=1}^m \prod_{e\in E} 
\exp\left(-\binom{O_j^{e}}{2} 
\left[\gamma_j^{e, O_j^{e}+1} - \gamma_j^{e, O_j^{e}}\right]\right)\times\prod_{\ell=1}^{I_j^{e}-O_j^{e}}
\exp\left(-\binom{\ell}{2} 
\left[\gamma_j^{e, \ell} - \gamma_j^{e, \ell-1}\right]\right),
\end{align*}
where, for gene $j$ and branch $e$, 
$I_j^{e}$ is the number of lineages entering $e$,
$O_j^{e}$ is the number of lineages exiting $e$, and
$\gamma_j^{e,\ell}$ is the $\ell^{th}$ coalescence time in $e$;
for convenience, we let $\gamma_j^{e,0}$ and
$\gamma_j^{e,I_j^{e}-O_j^{e}+1}$ be respectively
the divergence times (expressed in coalescence time units) of $e$ and of its parent population (which depend on $\Gamma$). 
We write $\{\tree_j\}_j \sim \mathcal{D}_{\mathrm{s}}^{m}[\stree,\Gamma]$ to indicate that the $m$ gene trees $\{\tree_j\}_j$ are independently
distributed according to the MSC on species tree $\stree,\Gamma$.
To be specific, $\tree_j$ is the \emph{unrooted gene tree topology}---without branch lengths---and we remark that,
under the MSC, $\tree_j$ is binary with probability $1$.
Throughout we assume that the $\tree_j$'s are known
and were reconstructed \emph{without estimation error}.

\paragraph{Internode distance} Assume we are given
$m$ gene trees $\{\tree_j\}_j$ over the $n$ species 
$\{1,\ldots,n\}$. For any pair of species $x,y$ and gene $j$,
we let $d^j_\mathrm{g}(x,y)$ be the \emph{graph distance} between $x$ and $y$ on $\tree_j$, i.e.~the number
of edges on the unique path between $x$ and $y$. 
The \emph{internode distance} between $x$ and $y$
is defined as the average graph distance across genes,
i.e.
\begin{align*}
\inode{m}(x,y)
=
\frac{1}{m} \sum_{j=1}^m d^{\tree_j}_\mathrm{g}(x,y).
\end{align*}
Under the MSC, the internode distances $\inode{m}(x,y))_{x,y}$ are correlated
random variables whose joint distribution depends
in the a complex way on the species tree $(\stree,\Gamma)$.
Here follows a remarkable fact about internode distance ~\cite{njst,kreidl2011,AlDeRh:18}.
Let $\einode(x,y)$ be the expectation of $\inode{m}(x,y)$
under the MSC and let $\ustree$ be the unrooted version
of the species tree $\stree$. Then $(\einode(x,y))_{x,y}$
is an additive metric associated\footnote{Note however that the associated branch lengths may differ from
	$\Gamma$.} to $\ustree$ (see e.g.~\cite{Steel:16}). In particular, whenever $\ustree$ 
restricted to species $x,y,w,z$ has quartet topology $xy|wz$
(i.e.~the middle edge of the restriction to $x,y,w,z$ splits
$x,y$ from $w,z$), it holds that\footnote{Note that it is trivial that $(d^{\tree_j}_\mathrm{g}(x,y))_{x,y}$ is an additive metric
associated to \emph{gene tree} $\tree_j$. On the other hand it is far from trivial that averaging over the MSC leads to an additive metric associated to the \emph{species tree}.}
\begin{align*}
\einode(x,w)+\einode(y,z) 
= \einode(x,z) + \einode(y,w)
\geq \einode(x,y) + \einode(w,z).
\end{align*} 
This result forms the basis for many popular multilocus reconstruction
methods, including NJst~\cite{njst} and ASTRID~\cite{astrid},
which apply standard distance-based methods to the
internode distances 
\begin{align*}
(\inode{m}(x,y))_{x,y}.
\end{align*}

\paragraph{Main results} 
By the law of large numbers, for all
pairs of species $x,y$
\begin{align*}
\inode{m}(x,y)
\to \einode(x,y),
\end{align*}
with probability $1$ as $m \to +\infty$, a fact that can be used
to establish the statistical consistency (i.e.~the guarantee that
the true specie tree is recovered as long as $m$ is large enough) of internode distance-based methods such as NJst~\cite{njst}.
However, as far as we know, nothing is known about the sample
complexity of internode distance-based methods, i.e.~how many
genes are needed to reconstruct the species tree with high probability---say $99\%$---as a function of some structural
properties of the species tree---primarily the number of species $n$ and the shortest branch length $f$? We do not answer this important but technically difficult question here, but we make progress towards its resolution by providing a lower bound on the worst-case
variance of internode distance.
Let $d_\mathrm{g}^{\ustree}(x,y)$ denote the graph distance
between $x$ and $y$ on $\ustree$. 
\begin{thm}[Lower bound on the worst-case variance of internode distance]\label{thm:1}
There exists a constant $C > 0$ such that, 
for any integer $n \geq 4$ and real $f > 0$, 
there is a
species tree $(\stree,\Gamma)$ with $n$ leaves 
and shortest branch length $f$ such that the following holds: for all pairs of species $\ell,\ell'$ and all
integers $m \geq 1$, if $\{\tree_j\}_j \sim \mathcal{D}_{\mathrm{s}}^{m}[\stree,\Gamma]$
then
\begin{align}\label{eq:main-top}
\var\left[\inode{m}(\ell,\ell')\right] \geq C \frac{d_\mathrm{g}^{\ustree}(\ell,\ell')}{m},
\end{align}
and, furthermore,
\begin{align}\label{eq:main-max}
\max_{\ell,\ell'} \var\left[\inode{m}(\ell,\ell')\right] \geq C \frac{n}{m},
\end{align}
\end{thm}
\noindent In words, there are species trees for which the variance of internode distance scales as the graph distance---which can be of order $n$---divided by $m$. 
The proof of Theorem~\ref{thm:1} is detailed in
Section~\ref{sec:bad}.

\section{Discussion}
\label{sec:discussion}

How is Theorem~\ref{thm:1}
related to the sample complexity of species tree estimation
methods?
The natural approach for deriving bounds on the number of genes
required for high-probability reconstruction in distance-based 
methods is to show that the estimated distances used are sufficiently
concentrated around their expectations---provided that $m$ is large enough as a function of
$n$ and $f$ (e.g.~\cite{ErStSzWa:99a,Atteson:99}; but see~\cite{Roch:10} for a more refined analysis). 
In particular, one needs to \emph{control the variance} of distance estimates.

\paragraph{Practical implications}
 Bound~\eqref{eq:main-max} in Theorem~\ref{thm:1} implies that
to make \emph{all} variances negligible the number of genes
$m$ is required to scale at least linearly 
in the number of species $n$. In contrast,
certain quartet-based methods such as ASTRAL~\cite{astral,astral-2} have a sample complexity 
scaling only logarithmically in $n$~\cite{ShRoMi:18}.

On the other hand, Bound~\eqref{eq:main-max} is only relevant
for those reconstruction algorithms using \emph{all} distances,  for instance NJst which is based on Neighbor-Joining~\cite{Atteson:99,LaceyChang:06}. 
Many so-called fast-converging reconstruction methods purposely use only a \emph{strict subset} of all distances,
specifically those distances within a constant factor of
the ``depth'' of the species tree. Refer to~\cite{ErStSzWa:99a} 
for a formal definition of the depth, but for our purposes it will suffice
to note that in the case of graph distance the
depth is at most of the order of $\log n$. Hence Bound~\eqref{eq:main-top} suggests 
it may still possible to achieve a sample complexity comparable
to that of ASTRAL---if one uses a fast-converging method (within ASTRID for instance). 

\paragraph{The impact of correlation}  
Theorem~\ref{thm:1} does not in fact lead to a bound on the
sample complexity of internode distance-based reconstruction
methods. For one, Theorem~\ref{thm:1} only gives a lower bound
	on the variance. One may be able to construct examples
	where the variance is even larger. In general, analyzing the
	behavior of internode distance is quite challenging because
	it depends on the \emph{full} multispecies coalescent process
	in a rather tangled manner.
	
	Perhaps more importantly, 
	the variance itself is not enough
	to obtain tight bounds on the sample complexity. One problem
	is \emph{correlation}. Because $\inode{m}(x,y)$ and $\inode{m}(w,z)$
	are obtained using the same gene trees, they are highly
	correlated random variables. One should expect this correlation to produce
	cancellations (e.g.~in the four-point condition; see~\cite{Steel:16}) that could drastically lower the sample complexity. The importance
	of this effect remains to be studied.

\paragraph{Gene tree estimation error}
We pointed out above that quartet-based methods such as ASTRAL may be less sensitive to long distances than 
internode distance-based methods such as NJst. 
An important caveat is the assumption that gene trees
are \emph{perfectly reconstructed}. In reality, gene tree estimation
errors are likely common and are also affected
by long distances (see e.g.~\cite{ErStSzWa:99a}).
A more satisfactory approach would account for these
errors or would consider simultaneously sequence-length
and gene-number requirements. Few such analyses
have so far been performed because of technical challenges~\cite{RochWarnow2015,DaNoRo:15,MosselRoch:15,DaMoNoRo:17}.

\section{Variance of internode distance}
\label{sec:bad}

In this section, we prove Theorem~\ref{thm:1}.
Our analysis of internode distance is based on the
construction of a special species tree where its variance
is easier to control. We begin with a high-level proof
sketch:
\begin{itemize}
	\item Our special example
	is a caterpillar tree with an alternation of \emph{short}
	and \emph{long} branches along the backbone.
	
	\item The \emph{short}
	branches produce ``local uncertainty'' in the number
	of lineages that coalesce onto the path between two fixed leaves.
	The \emph{long} branches make these contributions to the internode distance ``roughly
	independent'' along the backbone.
	
	\item As a result, the internode distance is, up to a small
	error, a sum of independent and identically distributed
	contributions. Hence, its variance grows linearly with graph distance.
	
\end{itemize}

\paragraph{Setting for analysis}
We fix the number of species $n$ and we assume
for convenience that $n$ is even.\footnote{A straightforward modification of the argument also works for odd $n$.}
Recall also that $f$ will denote the length of the shortest
branch in coalescent time units. We consider the
species tree $(\stree, \Gamma)$ depicted in Figure~\ref{fig:caterpillar}.
\begin{figure}
	\centering
	\includegraphics[width = 0.8\textwidth]{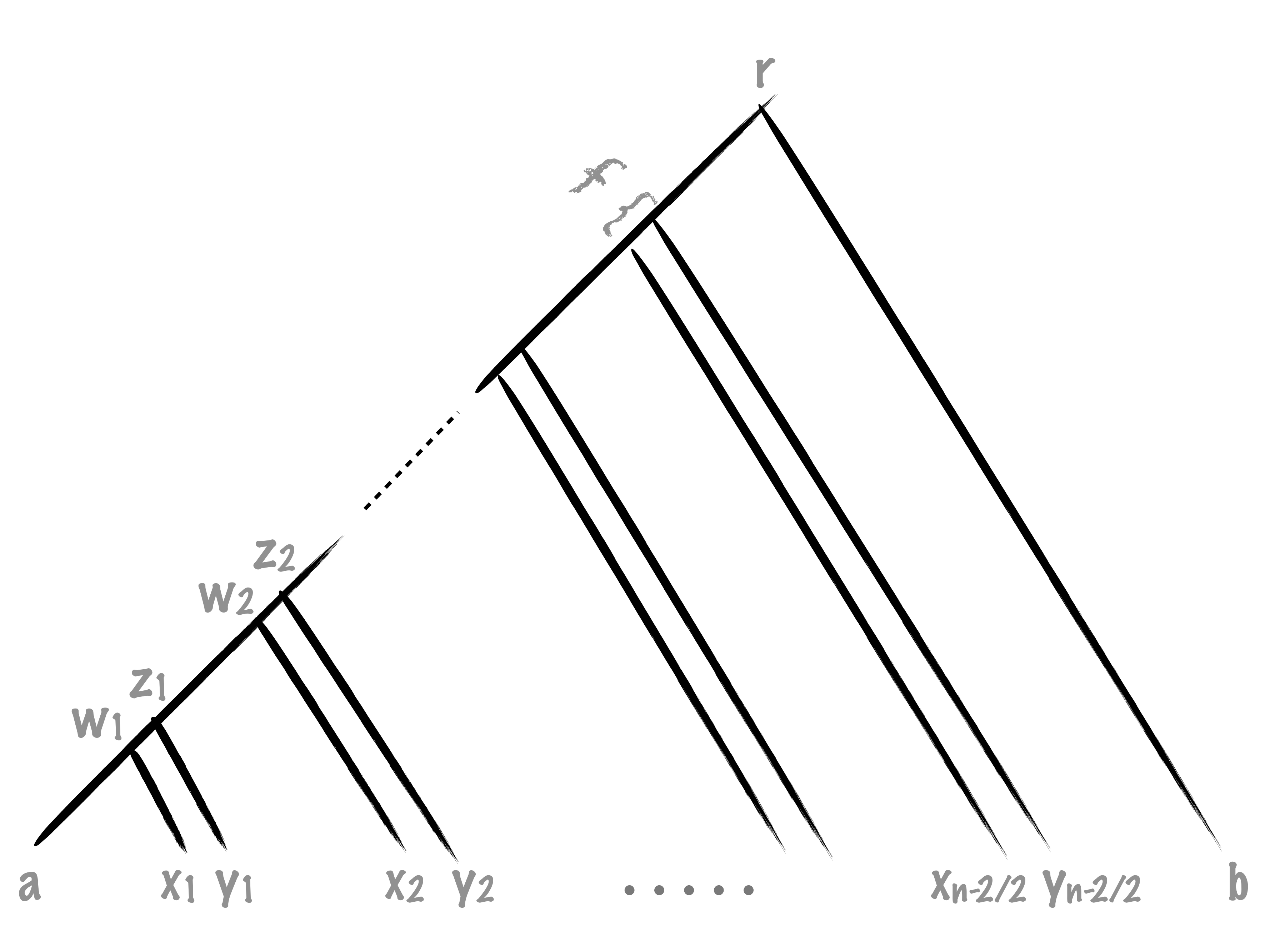}
	\caption{The species tree used in the analysis.}\label{fig:caterpillar}
\end{figure}
Specifically, $\stree$ is a caterpillar tree: its \emph{backbone}
is an $n-1$-edge path 
$$(a,w_1), (w_1,z_1), (z_1,w_2), (w_2,z_2), \ldots, (w_{\frac{n-2}{2}},z_{\frac{n-2}{2}}),(z_{\frac{n-2}{2}},r)$$
connecting leaf $a$ to root $r = w_{n/2}$; each vertex $w_i$
on the backbone is incident with an edge $(w_i,x_i)$
to leaf $x_i$; each vertex $z_i$
on the backbone is incident with an edge $(z_i,y_i)$
to leaf $y_i$; root $r$
is incident with an edge $(r,b)$
to leaf $b$. Each edge of the form $e = (w_i,z_i)$ is
a \emph{short} edge of length $\Gamma_e = f$, while
all other edges are \emph{long} edges of length $g = 4 \log n$.

\paragraph{Proof of Theorem~\ref{thm:1}}
Recall that our goal is to prove that for all pairs of species $\ell,\ell'$ and all
integers $m \geq 1$, if $\{\tree_j\}_j \sim \mathcal{D}_{\mathrm{s}}^{m}[\stree,\Gamma]$
then
\begin{align*}\label{eq:main-top-repeated}
\var\left[\inode{m}(\ell,\ell')\right] \geq C \frac{d_\mathrm{g}^{\ustree}(\ell,\ell')}{m}.
\end{align*}
To simplify the analysis, we detail the argument
in the case $\ell = a$ and $\ell' = b$ only. The other cases
follow similarly. 

We first reduce the computation to a single gene.
Recall that
\begin{align*}
\inode{m}(a,b)
=
\frac{1}{m} \sum_{j=1}^m d^{\tree_j}_\mathrm{g}(a,b).
\end{align*}
\begin{lemma}[Reduction to a single gene]\label{claim:single-gene}
	For any $m$, it holds that
	\begin{align*}
	\var\left[\inode{m}(\ell,\ell')\right] 
	=
	\frac{1}{m}
	\var\left[d^{\tree_1}_\mathrm{g}(a,b)\right].
	\end{align*}
\end{lemma}
\begin{proof}
	Because the $\tree_j$'s are independent and identically
	distributed, it follows that
	\begin{align*}
	\var\left[\inode{m}(\ell,\ell')\right] 
	&= \var\left[\frac{1}{m} \sum_{j=1}^m d^{\tree_j}_\mathrm{g}(a,b)\right]
	=  \frac{1}{m^2}  \sum_{j=1}^m  \var\left[ d^{\tree_j}_\mathrm{g}(a,b)\right]
	=
	\frac{1}{m}
	\var\left[d^{\tree_1}_\mathrm{g}(a,b)\right],
	\end{align*}
	as claimed.
\end{proof}

We refer to the $2$-edge path $\{(w_i,z_i), (z_i,w_{i+1})\}$ as the \emph{$i$-th block}.
The purpose of the long backbone
edges is to create independence between the contributions
of the blocks. To make that explicit, let $\fcoal_i$ be
the event that, in $\tree_1$, all lineages entering the
edge $(z_i,w_{i+1})$ have coalesced by the end of the
edge (backwards in time). And let $\fcoal = \cap_i \fcoal_i$.
\begin{lemma}[Full coalescence on all blocks]\label{claim:full-coal}
	It holds that
	\begin{align*}
	\prob[\fcoal] \geq 1 - 1/n.
	\end{align*}
\end{lemma}
\begin{proof}
	By the multiplication rule and the fact that $\fcoal_i$ only
	depends on the number of lineages entering $(w_i,z_i)$,
	we have
	\begin{align*}
	\prob[\fcoal] 
	&= \prod_i \prob[\fcoal_i\,|\,\fcoal_{1} \cap \cdots \cap \fcoal_{i-1}]
	= \left(\prob[\fcoal_1]\right)^{n/2-1}
	\geq 1 - (n/2-1) \left(1- \prob[\fcoal_1]\right).
	\end{align*}
	It remains to upper bound $\prob[\fcoal_1^c]$.
	We have either $2$ or $3$ lineages entering $(z_1,w_2)$. 
	In the former case, the failure to coalesce has probability
	$e^{-g}$, i.e.~the probability that an exponential with rate $1$ is greater than $g$. In the latter case, the failure to fully coalesce
	has probability at most $e^{-3(g/2)} + e^{-g/2}$, i.e.~the probability that either the first coalescence (happening at rate $3$) or the second one (happening at rate $1$) takes more than $g/2$. Either way this gives at most $\prob[\fcoal_1^c] \leq 2 e^{-g/2}$. With $g = 4 \log n = 2 \log n^2$ above, we get the claim.
\end{proof}

We now control the contribution from each block. Let
$X_i$ be the number of lineages coalescing into the
path between $a$ and $b$ \emph{on the $i$-th block}.
Conditioning on $\fcoal$, we have $X_i \in \{1,2\}$ 
and we have further that all $X_i$'s are independent and identically distributed. This leads to the following
bound.
\begin{lemma}[Linear variance]\label{claim:linear-var}
	It holds that
	\begin{align*}
	\var\left[d^{\tree_1}_\mathrm{g}(a,b)\right]
	\geq \frac{n-2}{2} \,\var\left[X_1\middle| \fcoal_1\right]\,\prob[\fcoal].
	\end{align*}
\end{lemma}
\begin{proof}
	By the conditional variance formula, letting $\mathbbm{1}_\fcoal$ be the indicator of $\fcoal$,
	\begin{align*}
	\var\left[d^{\tree_1}_\mathrm{g}(a,b)\right]
	&\geq \expec\left[
	\var\left[d^{\tree_1}_\mathrm{g}(a,b)\middle| \mathbbm{1}_\fcoal\right]
	\right]
	\geq \var\left[d^{\tree_1}_\mathrm{g}(a,b)\middle| \fcoal\right]
	\prob[\fcoal].
	\end{align*}
On the event $\fcoal$, it holds that
	\begin{align*}
	d^{\tree_1}_\mathrm{g}(a,b)
	&= \sum_{i} X_i.
	\end{align*}
	Moreover, conditioning on $\fcoal$ makes the $X_i$'s independent 
	and identically distributed. Hence we have finally
	\begin{align*}
	\var\left[d^{\tree_1}_\mathrm{g}(a,b)\right]
	&\geq 
	\frac{n-2}{2} \,\var\left[X_1 \middle| \fcoal\right]
	\,\prob[\fcoal]
	\geq \frac{n-2}{2} \,\var\left[X_1 \middle| \fcoal_1\right]
	\,\prob[\fcoal],
	\end{align*}
	where we used the fact that $X_1$
	depends on $\fcoal$ only through $\fcoal_1$.
\end{proof}

The final step is to bound the contribution to the variance
from a single block.
\begin{lemma}[Contribution from a block]\label{claim:single-block}
	It holds that
	\begin{align*}
	\var\left[X_1\middle| \fcoal_1\right]
	= \frac{1}{3}e^{-f}\left(1 - \frac{1}{3}e^{-f}\right) = \frac{2}{9}\left(1 - \Theta(f)\right),
	\end{align*}
	for $f$ small, where we used the standard Big-Theta notation. 
\end{lemma}
\begin{proof}
	As we pointed out earlier, conditioning on $\fcoal_1$,
	we have $X_1 \in \{1,2\}$. In particular $X_1 - 1$ is a Bernoulli 
	random variable whose variance $\prob\left[X_1 - 1 = 1\middle| \fcoal_1\right](1- \prob\left[X_1 - 1 = 1\middle| \fcoal_1\right])$ is the same as the variance of $X_1$ itself. So we need to compute the probability that $X_1 = 2$, conditioned on $\fcoal_1$. There are four scenarios to consider (depending on whether or not there is coalescence in the short branch $(w_1,z_1)$ and which
	coalescence occurs first in the long branch $(z_1,w_2)$), only one of which produces $X_1 = 1$:
	\begin{itemize}
		\item No coalescence occurs in $(w_1,z_1)$ and the first coalescence in $(z_1,w_2)$ is between the lineages coming from $x_1$ and $y_1$. This event has probability $\frac{1}{3}e^{-f}$ by symmetry when conditioning on $\fcoal_1$.
	\end{itemize}
	Hence $\prob\left[X_1 = 2\middle| \fcoal_1\right] = 1 - \frac{1}{3}e^{-f}$.
\end{proof}

By combining Lemmas~\ref{claim:single-gene}, \ref{claim:full-coal}, \ref{claim:linear-var} and \ref{claim:single-block}, we get that
\begin{align*}
\var\left[\inode{m}(\ell,\ell')\right] 
\geq
\frac{1}{m} \times \frac{n-2}{2} \times \frac{1}{3}e^{-f}\left(1 - \frac{1}{3}e^{-f}\right)  \times \left(1 - \frac{1}{n}\right).
\end{align*}
Choosing $C$ small enough concludes the proof of the theorem.

\section{Conclusion}

To summarize, we have derived a new lower bound
on the worst-case variance of internode distance under the multispecies
coalescent. No such bounds were previously known as
far as we know.
Our results suggest it may be preferable to use
fast-converging methods when working with internode distances for species tree estimation. 
The problem of providing
tight upper bounds on the sample complexity
of internode distance-based methods remains however
an important open question.


\section*{Acknowledgments}

This work was supported by funding from the U.S. National Science Foundation  
DMS-1149312 (CAREER), DMS-1614242 and CCF-1740707 (TRIPODS).
We thank Tandy Warnow for suggesting the problem and for helpful discussions.

\bibliographystyle{plain}
\bibliography{nute-roch-warnow}

%
%

\end{document}